\documentclass[a4paper,UKenglish,cleveref, autoref, thm-restate]{lipics-v2021}

\usepackage[ruled,vlined,linesnumbered,commentsnumbered]{algorithm2e}
\usepackage{xcolor}
\usepackage{hyperref}
\usepackage{algcompatible}
\usepackage{amsmath,amssymb,amsfonts}
\usepackage{graphicx}

\nolinenumbers

\title             {Bicriteria approximation for {\em k}-edge-connectivity}
\titlerunning{Bicriteria approximation for {\em k}-edge-connectivity}

\author{Zeev Nutov}{The Open University of Israel}{nutov@openu.ac.il}{https://orcid.org/0000-0002-6629-3243}{}
\author{Reut Cohen}{The Open University of Israel}{reut6511@gmail.com}{}{}

\authorrunning{Zeev Nutov and Reut Cohen}
\Copyright{Zeev Nutov and Reut Cohen}

\ccsdesc[100]{Theory of computation~Design and analysis of algorithms}

\begin{document}

\maketitle

\newcommand {\ignore} [1] {}

% \newtheorem{fact}[lemma]{Fact}
% \newtheorem{assumption}[lemma]{Assumption}

% Calligraphic 
\def\CC  {{\cal C}}
\def\LL  {{\cal L}}
\def\RR  {{\cal R}}
\def\TT  {{\cal T}}

% Greek letters
\def\de  {\delta}
\def\be  {\beta}
\def\th   {\theta}
\def\De  {\Delta}

% Sets
\def\empt {\emptyset}
\def\sem  {\setminus}
\def\subs {\subseteq}

% For brevity
\def\f  {\frac}
\def\opt {\sf opt}
\def\span {\sf span}

% Problems
\def\kECSS  {{\sc $k$-ECSS}}  % $k$-Edge Connected Spanning Subgraph
\def\kECSM {{\sc $k$-ECSM}}  % $k$-Edge Connected Spanning Multisubgraph
\def\kCSS {{\sc $k$-CSS}}        % $k$-Connected Spanning Subgraph

\keywords{$k$-edge-connected subgraph, bicriteria approximation, iterative LP-rounding}

\begin{abstract}
In the {\sc  $k$-Edge Connected Spanning Subgraph} ({\kECSS}) problem 
we are given a (multi-)graph $G=(V,E)$ with edge costs and an integer $k$, 
and seek a min-cost $k$-edge-connected spanning subgraph of $G$. 
The problem admits a $2$-approximation algorithm and no better approximation ratio is known.
Recently, Hershkowitz, Klein, and Zenklusen [STOC 24] gave a bicriteria $(1,k-10)$-approximation algorithm 
that computes a $(k-10)$-edge-connected spanning subgraph 
of cost at most the optimal value of a standard Cut-LP for {\kECSS}. 
We improve the bicriteria approximation to $(1,k-4)$ 
and also give another non-trivial bicriteria approximation $(3/2,k-2)$.

The {\sc $k$-Edge-Connected Spanning Multi-subgraph} ({\kECSM}) problem is almost the same as {\kECSS},
except that any edge can be selected multiple times at the same cost.
A $(1,k-p)$ bicriteria approximation for {\kECSS} w.r.t. Cut-LP implies 
approximation ratio  $1+p/k$ for {\kECSM}, hence our result also improves 
the approximation ratio for {\kECSM}.
\end{abstract}

%%%%%%%%%%%%%%%%%%%%
\section{Introduction} \label{s:intro}
%%%%%%%%%%%%%%%%%%%%

A graph is {\bf $k$-edge-connected} if it contains $k$ edge-disjoint paths between any two nodes.
We consider the following problem.

\begin{center} \fbox{\begin{minipage}{0.98\textwidth}
\underline{\sc  $k$-Edge-Connected Spanning Subgraph} ({\kECSS}) \\
{\em Input:} \ \ A (multi-)graph $G=(V,E)$ with edge costs and an integer $k$. \\
{\em Output:} A  min-cost $k$-edge-connected spanning subgraph $J$ of $G$. 
\end{minipage}} \end{center}

The problems admits approximation ratio $2$ by a reduction to a bidirected graph (c.f. \cite{K}), 
and no better approximation ratio is known.
Recently, Hershkowitz, Klein, and Zenklusen \cite{HKZ} (henceforth HKZ) 
gave a bicriteria $(1,k-10)$-approximation algorithm 
that computes a $(k-10)$-edge-connected spanning subgraph of cost
at most the optimal value of a standard Cut-LP for {\kECSS}.  
Our first result improves the connectivity approximation to $k-4$.

\begin{theorem} \label{t:main}
{\kECSS} admits a bicriteria approximation ratio $(1,k-4)$ w.r.t. the Cut-LP.
\end{theorem}

The {\sc $k$-Edge-Connected Spanning Multi-subgraph} ({\kECSM}) problem 
is almost the same as {\sc $k$-ECSS}, except that any edge can be selected multiple times at the same cost.
For {\kECSM}, approximation ratios $3/2$ for $k$ even and $3/2+O(1/k)$ for $k$ odd 
were known long time ago \cite{FJ1,FJ2}, and it was conjectured by Pritchard \cite{P} that 
{\kECSM} admits approximation ratio $1+O(1/k)$. 
HKZ \cite{HKZ} showed that {\kECSM} has an $\left(1+\Omega(1/k)\right)$-approximation threshold 
and that a $(1,k-p)$ bicriteria approximation for {\kECSS} 
w.r.t. the standard Cut-LP implies approximation ratio  $1+p/k$ for {\kECSM}. 
Thus the novel result of HKZ \cite{HKZ} implies an approximation ratio $1+10/k$ for {\kECSM}
(improving the previous ratio $1+O(1/\sqrt{k})$ of  \cite{KKGZ}),
and we reduce the approximation ratio to $1+4/k$.

\begin{corollary}
{\kECSM} admits approximation ratio $1+4/k$.
\end{corollary}

In addition, we also study bicriteria approximation ratios when 
the cost approximation is bellow $2$, and prove the following.

\begin{theorem} \label{t:main'}
{\kECSS} admits a bicriteria approximation ratio $(3/2,k-2)$.
\end{theorem}

Our paper leaves an open question whether for some constant $\epsilon >0$ 
it is possible to achieve a bicriteria approximation ratio $(2-\epsilon, k-1)$.
We note that while HKZ presented a novel original approach for attacking the problem
and gave the first additive bicriteria approximation,
our contribution here is more modest, as follows.
\begin{itemize}
\item
We present a simpler algorithm that uses the HKZ \cite{HKZ} idea and a much simpler proof, 
based on new properties of {\kECSS} Cut-LP extreme point solutions; see Lemmas \ref{l:main}.
\item
This enabled us to improve the connectivity approximation from $k-10$ to $k-4$. 
\item
We use the same idea to present another nontrivial bicriteria approximation $(3/2,k-2)$.
\end{itemize}

We now survey some related work.
{\kECSS} is an old fundamental problem in combina\-torial optimization and network design.
The case $k=1$ is the {\sc MST} problem, while for $k = 2$ the problem is 
MAX-SNP hard even for unit costs \cite{F}.
Pritchard \cite{P} showed that  there is a constant $\epsilon > 0$ such that for any $k \ge 2$,
{\kECSS} admits no $(1+\epsilon)$-approximation algorithm, unless P = NP.
No such hardness of approximation was known for {\kECSM},
and Pritchard \cite{P} conjectured that {\kECSM} admits approximation ratio $1+O(1/k)$.
This conjecture was resolved by HKZ \cite{HKZ}, 
who also proved a matching hardness of approximation result.

The {\kECSS} problem and its special cases have been extensively
studied, c.f \cite{Lov,F-aug,FJ1,FJ2,KV,CT,CL,K,GGTW, GG,P, SV,TZ,HKZ} for only a small sample of papers in the field. 
Nevertheless, the current best known approximation ratio for {\kECSS} is $2$, the same as was four decades ago. 
The $2$-approximation is obtained by bidirecting the edges of $G$, 
computing a min-cost directed spanning subgraph that contains $k$ edge-disjoin dipaths 
to all nodes from some root node (this can be done in polynomial time \cite{E}), 
and returning its underlying graph.
A $2$-approximation can also be achieved with the iterative rounding method \cite{J},
that gives a $2$-approximation also for the more general {\sc Steiner Network} problem, 
where we seek a min-cost subgraph that contains 
$r_{uv}$ edge-disjoint paths for every pair of nodes $u,v \in V$.
The directed version of {\kECSS} was also studied, c.f. \cite{E,GGTW,CT,K}, 
and for this case also no better than $2$-approximation is known, except for special cases. 

Bicriteria approximation algorithms were widely studied for 
degree constrained network design problems such as {\sc Bounded Degree MST} and {\sc Bounded Degree Steiner Network},
c.f. \cite{SL,LZ,LRS} and the references therein.
A particular case of a bicriteria approximation is when only one parameter is relaxed
while the cost/size of the solution is bounded by a budget $B$.
For example, the $(1-1/e)$-approximation for the {\sc Budgeted Max-Coverage} problem 
can be viewed as a bicriteria approximation $(1, 1-1/e)$ for {\sc Set Cover}, 
where only a fraction of $1-1/e$ of the elements is covered by ${\sf opt}$ sets.  
% where ${\sf opt}_{SC}$ is the optimal {\sc Set cover} solution value.
Similarly, in the budgeted version {\sc Budgeted ECSS} of {\kECSS}, instead of $k$ we are given a budget $B$ 
and seek a spanning subgraph of cost at most $B$ that has maximum edge-connectivity $k^*$.

One can view  Theorem~\ref{t:main} as a $k^*-4$ additive approximation for {\sc Budgeted ECSS}, where 
$k^*$ is the maximum edge-connectivity under the budget $B$.
We note that budgeted connectivity problems were studied before for unit costs.
Specifically, in the {\sc Min-Size \kECSS} and {\sc Min-Size \kCSS} problems 
we seek a $k$-edge-connected and $k$-node-connected, respectively, spanning subgraph with a minimal number of edges. 
Nutov \cite{N-small} showed that the following simple heuristic (previously analyzed by Cheriyan and Thurimella \cite{CT}) 
is a $(1,k-1)$ bicriteria approximation for {\sc Min-Size \kCSS}:
compute a min-size $(k-2)$-edge-cover $I_{k-2}$ (a spanning subgraph of minimal degree $\ge k-2$) 
and then augment it by an inclusion minimal edge set $F \subs E$ such that $I_{k-2} \cup F$ is $(k-1)$-connected. 
This $(1,k-1)$ bicriteria approximation implies a $k^*-1$ approximation for {\sc Budgeted Min-Size \kCSS},
which is tight, since the problem is NP-hard.
A similar result for {\kECSS}, for large enough $k$, can be deduced from the work of Gabow and Gallagher \cite{GG}. 

One can obtain a bicriteria approximation $(2/\th,\lfloor k/\th \rfloor)$  for {\kECSS} for any $\th \ge 1$,
by just solving the {\sc $\lfloor k/\th \rfloor$-ECSS} problem. 
This is since the $2$ approximation is w.r.t. the Cut-LP. 
Specifically, if $x$ is a {\kECSS} Cut-LP solution, then $x/\th$ is a  {\sc $\lfloor k/\th \rfloor$-ECSS} Cut-LP solution,
thus a $2$-approximation w.r.t. Cut-LP for {\sc $\lfloor k/\th \rfloor$-ECSS} computes a $\lfloor k/\th \rfloor$-connected subgraph
of cost $\le 2/\th \cdot {\sf opt}_k$.
Similar ``LP-scaling'' gives a bicriteria approximation $(2/\th, \lfloor r_{uv}/\th \rfloor)$ for any $\th \ge 1$
for the {\sc Steiner Network} problem. 

The node connectivity version {\kCSS} of {\kECSS} was also studied extensively 
\cite{KN,FL,ZN-comb,CT, N-small,CV,ZN-4,N-TCS,N-book};
it admits approximation ratio $4+\epsilon$ when $k$ is bounded by a constant, but for general $k$ 
only a polylogarithmic approximation is known \cite{FL,ZN-comb}. 
The {\sc Survivable Network Design Problem} ({\sc SNDP}) is the node-connectivity version of {\sc Steiner Network}
when there should be $r_{uv}$ internally node disjoint path between every $u,v \in V$;
{\kCSS} is a particular case when $r_{uv}=k$ for all $u,v \in V$.
The {\sc Rooted SNDP} is another particular case of {\sc SNDP} where we require $k$ disjoint paths from a root $s$ to 
every node in a given set $T$ of terminals, and the {\sc $k$-Out-Connected Spanning Subgraph} ({\sc $k$-OCSS}) problem 
is a particular case of {\sc Rooted SNDP} when $T=V \sem \{s\}$.
Since the best approximation ratios for these problems are w.r.t. to the Cut-LP (a.k.a. ``biset LP'') they can be used 
to obtain bicretiria approximations using the method described above.
We summarize the best approximation ratios for these problems 
and the bicriteria approximation ratios that can be derived for them 
by a simple ``LP-scaling'' in the following table. 

\begin{table}[htbp]  \label{tbl:cost} 
\begin{center} 
\begin{tabular}{|l|c|c|c|c|c|} 
\hline
{\bf problem}                          & {\sc SNDP}                              & {\sc Rooted SNDP}                   & {\sc $k$-CSS}                                                                   & {\sc $k$-OCSS}                                    
\\\hline  
{\bf approximation}                & $O(k^3 \log n)$ \cite{CK}      & $O(k\log k)$ \cite{N-rooted}      & $4+\f{4\lg k+1}{\lg n-\lg k}$ \cite{ZN-4}                           & $2$ \cite{FT}        
\\\hline
{\bf bicriteria appr.}                & $\f{1}{\th} O((k/\th)^3 \log n)$ & $\f{1}{\th} O((k/\th) \log (k/\th))$ & $\f{1}{\th} \left(4+\f{4\lg (k/\th)+1}{\lg n-\lg (k/\th)}\right)$ & $2/\th$         
\\\hline
\end{tabular}
\end{center}
\caption{LP and bicriteria approximation ratios for node connectivity problems.
The third row gives the approximation cost that can achieve connectivity $\lfloor k/\th \rfloor$ 
($\lfloor r_{uv}/\th \rfloor$ for {\sc SNDP}) for any $\th \ge 1$.}
\end{table}

\vspace*{-0.4cm}
This paper is organized as follows. 
In Sect.~\ref{s:algo} we describe the iterative relaxation algorithm.
In Sect.~\ref{s:cut} we prove the connectivity guarantee, and 
in Sect.~\ref{s:main} we prove a certain property of extreme point solutions that is used in the algorithm.
Sect.~\ref{s:main'} gives a bicriteria approximation $(3/2, k-2)$.

%%%%%%%%%%%%%%%%%%%%%%%%%%%%%%%%%%%
\section{The Algorithm} \label{s:algo}
%%%%%%%%%%%%%%%%%%%%%%%%%%%%%%%%%%%

For an edge set $F$ and disjoint node sets $S,T$ let $\de_F(S,T)$ denote the set of edges in $F$ with 
one end in $S$ and the other end in $T$, and let $d_F(S,T)=|\de_F(S,T)|$.
Let $\de_F(S)=\de_F(S,\bar{S})$ and $d_F(S)=|\de_F(S)|$,
where $\bar{S}=V \sem S$ is the node complement of $S$. 
% The default subscript is $E$, e.g., $\de(S)=\de_E(S)$.
For $x \in \mathbb{R}^E$ let $x(F)=\sum_{e \in F} x(e)$. % and $d_x(S)=\sum_{e \in \de(S)} x_e$.
We will often refer to a proper subset $S$ of $V$ (so $\empt \ne S \subset V$) as a {\bf cut}. 
A standard Cut-LP for {\kECSS} is:
\[ \displaystyle
\begin{array} {lllllll} 
& \min              & \displaystyle c^T \cdot x & \\
& \ \mbox{s.t.} & x(\de_E(S)) \ge k                & \forall \empt \ne S \subset V \\
&                      & 0 \le x_e \le 1                  & \forall e \in E                                                                 
\end{array}
\]

As in HKZ \cite{HKZ}, we will use the iterative relaxation method 
(that previously was used for the {\sc Degree Bounded MST} problem \cite{SL}).
This means that while $E \ne \empt$, 
we repeatedly compute an extreme point optimal solution $x$ to the Cut-LP above, 
and do at least one of the following steps:
\begin{enumerate}
\item
Remove from $E$ an edge $e$ with $x_e=0$.
\item
Add to a partial solution $I$ an edge $e$ with $x_e=1$, and remove $e$ from $E$.
\item
Relax some cut constraints $x(\de(S))  \ge k$ to $x(\de(S)) \ge k-q$ for some integer $q$.  
\end{enumerate}

The relaxation of the cut constraints will be implemented in two different ways.
\begin{enumerate}[(i)]
\item
Relaxing the cut constraint of a single cut $S$ (and also of $\bar{S}=V \sem S$)
and then contracting $S$ into a single node $v_S$;
this also removes the constraint of every cut $A$ such that both $A \cap S$ and $A \cap \bar{S}$ are non-empty.  
We will store the nodes that correspond to such relaxed contracted cuts $S$ in a set $U$.
\item
Relaxing the constraints of all cuts $S$ that separate two nodes $u_1,u_2 \in U$ 
(that correspond to contracted sets $C_1,C_2$) to $x(\de(S)) \ge k-2$.
This is achieved by adding to the partial solution a ``ghost edge'' $u_1u_2$ of capacity $2$,
and considering the ``residual'' demands of sets w.r.t the ghost edges.  
We will store the added ghost edges in a set $H$.
\end{enumerate}

Instead of working with the cut constraints $x(\de_E(S)) \ge k$ it would be convenient to consider 
the residual constraints $x(\de_E(S)) \ge f(S)$ for an appropriately defined set function~$f$.
To define this $f$, suppose that we are already given:
\begin{itemize}
\item
A partial solution $I$ of already chosen ``integral'' edges, that were removed from $E$. 
\item
A set $H$ of ``ghost-edges'' -- these are new virtual edges that are added to the solution.
\item
A set $U$ of nodes that correspond to contracted cuts whose constraints are relaxed. 
\end{itemize}
Let $f(S)$ be a set function defined on proper node subsets $S$ by
\begin{equation} \label{f}
f(S) = \left \{ \begin{array}{ll}
k-d_I(S)-2d_H(S)-2 \ \ & \mbox{if } S=\{u\} \mbox{ or } \bar{S}=\{u\} \mbox{ for some } u \in U \\
k-d_I(S)-2d_H(S)        & \mbox{otherwise}
\end{array} \right .
\end{equation}
This function $f(S)$ is a relaxation of the usual residual function $k-d_I(S)$ of {\kECSS}.
Consider the following LP-relaxation with the function $f$ above:
\begin{equation} \label{LP}
\displaystyle
\begin{array} {lllllll} 
& \min             & c^{T} \cdot x              &                                                  \\ 
& \ \mbox{s.t.} & x(\de_E(S)) \ge f(S) & \forall \empt \ne S \subset V  \\
&                      & 0 \le x_e \le 1           & \forall e \in E    
\end{array}
\end{equation}

A set $S$ is {\bf $f$-positive} if $f(S)>0$. 
Given an LP solution $x$ we say that $S$ is {\bf $x$-tight} if 
the LP-inequality of $S$ holds with equality, namely, if  $x(\de(S))=f(S)$. 
An {\bf $x$-core}, or simply a {\bf core} if $x$ is clear from the context,
is an inclusion-minimal $f$-positive $x$-tight set 
(note that an $x$-core must be $f$-positive and not only $x$-tight).

To {\bf contract a node subset $C$} of a graph $G=(V,E)$ means to identify all nodes in $C$ into one node $v_C$;
edges that have an end in $C$ now have $v_C$ as their end and 
the arising loops (if any) are deleted. 
During the algorithm, we iteratively contract certain node subsets of $G$, 
so we denote the initial graph by $G_0=(V_0,E_0)$.

Our algorithm is given in Algorithm~\ref{alg:main}. 
The main difference between Algorithm~\ref{alg:main} and that of HKZ is that we never add parallel ghost edges,
namely we must have $d_H(u,v)=0$ whenever we add a ghost edge $uv$. 
This has a chain reaction of simplifying some other parts of the algorithm. 
For example, we add a ghost edge when $d_I(u,v) \ge \lceil (k-3)/2 \rceil$ (and $d_H(u,v)=0$)
while the analogous condition in HKZ is $d_{I \cup H}(uv) \in \left[k/2-2,k/2 \right)$. 
The upper bound $d_{I \cup H}(u,v) < k/2$ in HKZ is needed to control the number of parallel ghost edges added, 
which also complicates the analysis.

In Section~\ref{s:main} we will prove the following lemma about extreme point solutions of LP (\ref{LP}).

\medskip 

\begin{algorithm}[H]
\caption{A bicriteria $(1,k-4)$-approximation} \label{alg:main}
{\bf initialization}: $I \gets \empt$, $H \gets \empt$, $U \gets \empt$, $\mu \gets \lceil (k-3)/2\rceil$ \\
\While{$E \ne \empt$}
{
compute an extreme point solution $x$ to LP  (\ref{LP}) with $f$ in (\ref{f}) \\
remove from $E$ every edge $e$ with $x_e=0$ \\
\If{\em there is $e \in E$ with $x_e =1$}
{move from $E$ to $I$ every edge $e$ with $x_e=1$ and goto line 2} 
\If{\em there is an $x$-core $C$ with $d_E(C) \in \{2,3\}$}
{contract $C$ into a single node $v_C$ and update $U \gets (U \sem C) \cup \{v_C\}$} 
\ElseIf{\em there are $u,v \in U$ with $d_I(u,v) \ge \mu$ and $d_H(u,v) =0$}  
{add to $H$ a ghost-edge $uv$ and  update $U \gets U \sem \{u,v\}$}
}
\Return{$J=I$}
\end{algorithm}

\medskip 

\begin{lemma} \label{l:main}
Let $f$ be defined by (\ref{f}) and suppose that
$d_H(u) \le 1$ and $d_I(u)+2d_H(u) \ge k-2$ for all $u \in U$.
Let $x$ be an extreme point of the polytope 
$P=\{x \in [0,1]^E:x(\de_E(S)) \ge f(S)\}$ % of the LP in (\ref{LP})
such that $0<x_e<1$ for all $e \in E$, where $E \ne \empt$. 
Then at least one of the following holds.
\begin{enumerate}[(i)]
\item
There is an $x$-core $C$ with $d_E(C) \in \{2,3\}$; moreover, no edge in $E$ has both ends in $C$.
\item
There are $u,v \in U$ with $d_H(u,v)=0$ and $d_I(u,v) \ge \mu$, where $\mu=\lceil (k-3)/2 \rceil$.
\end{enumerate}
Furthermore, such $u,v$ or such $C$ can be found in polynomial time.
\end{lemma}

Interestingly, the two cases in Lemma~\ref{l:main} depend on 
whether there exist a laminar set family $\LL$ such that $x$ is the unique solution
to the equation system $\{x(\de_E(S)=f(S):S \in \LL\}$. %  see precise definitions in Section~\ref{s:main}.
Specifically, if such $\LL$ exists then 
case (i) in Lemma~\ref{l:main} holds, otherwise case~(ii) must hold.

In the next Section~\ref{s:cut} we will show that the algorithm terminates after a polynomial number of iterations
and prove the connectivity guarantee $k-4$. 
Since in each iteration the cut constraints are only relaxed 
(by contractions, or by including nodes in $U$, or by adding ghost edges),
the cost of the produced solution is at most the initial LP-value.
In this context note that the $-2$ relaxation of the nodes $u,v \in U$ is not canceled when $u,v$ are removed 
from $U$, % at line~7, 
it is just replaced  by the $-2$ relaxation caused by the added ghost edge $uv$.

%%%%%%%%%%%%%%%%%%%%%%%%%
\section{Connectivity guarantee} \label{s:cut}
%%%%%%%%%%%%%%%%%%%%%%%%%

For every ghost edge $h=uv$ added, let $I(h)=I(uv)$ be the 
set of (at least $\mu=\lceil (k-3)/2 \rceil$) edges in $J$ that appear in $\de_I(u,v)$. 
A key observation is the following.

\begin{lemma} \label{l:h} 
For any distinct ghost edges $h,h'$ added during the algorithm, $I(h) \cap I(h')=\empt$.
\end{lemma}
\begin{proof}
This follows from the condition that when a ghost edge $h = uv$ added, 
% at lines 9,10 of the algorithm, 
there is no other ghost edge between $u$ and $v$.
Specifically, assume w.l.o.g. that $h=uv$ was added after $h'$. 
Then all edges in $I(h)$ are $uv$-edges -- go between the sets $C_u,C_v$ that correspond to $u,v$,
while no edge in $I(h')$ is a $uv$-edge. 
\end{proof}

\begin{lemma}\label{l:C}
When a core $C$ % at step $7$ of the algorithm 
is found in the algorithm but not yet contracted, 
$d_H(C) \le 1$, $x(\de_E(C)) \in \{1,2\}$, and $d_I(C) = k-2d_H(C)-x(\de_E(C))$. Consequently,
\begin{itemize}
\item
$d_I(C) \in \{k-1,k-2\}$ if $d_H(C)=0$. 
\item
$d_I(C) \in \{k-3,k-4\}$ if $d_H(C)=1$.
\end{itemize}
\end{lemma}
\begin{proof}
If $d_H(C) \ge 2$ then $d_I(C)+2d_H(C) \ge 2\mu +4 \ge k$, hence $f(C) \le 0$, 
contradicting that $C$ is a core (recall that a core must be an $f$-positive set).  
We have $d_E(C) \in \{2,3\}$, so there are $2$ or $3$ fractional edges in $\de_E(C)$. 
Since $C$ is $x$-tight, $x(\de_E(C))$ is a positive integer that is either $1$ or $2$,
so there are $k-1-2d_H(C)$ or $k-2-2d_H(C)$ (integral) edges in $\de_I(C)$. 
If $d_H(C)=0$ then $d_I(C) \in \{k-1,k-2\}$ and if $d_H(C)=1$ then $d_I(C) \in \{k-3,k-4\}$;
see Fig.~\ref{f:core}(a,b) where $\de_E(x(C))=2$ and $d_H(C)=0$ or $d_H(C)=1$, respectively. 
\end{proof}

\begin{figure} \centering \includegraphics[scale=0.5]{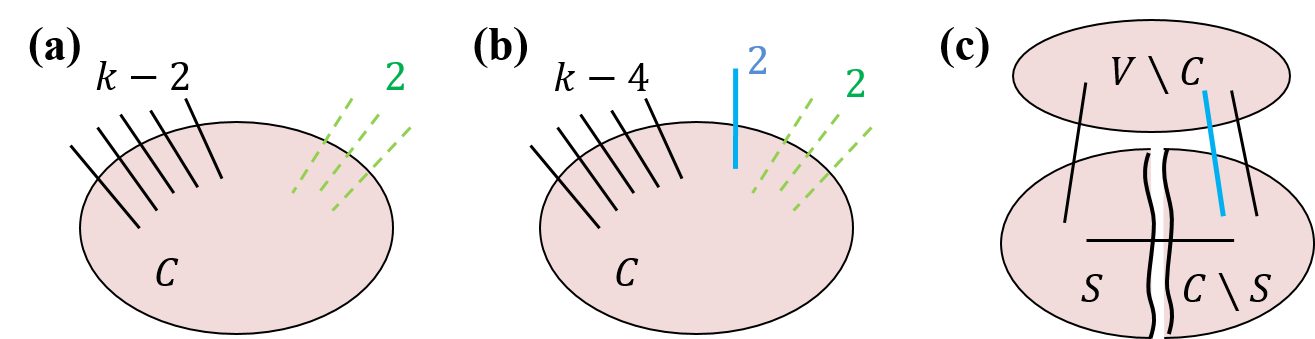}
\caption{Illustration to the proofs of Lemmas \ref{l:C} and \ref{l:cut}.
Edges in $I$ are shown by black lines, 
edges in $E$ (the fractional edges) by green dashed lines, 
and the ghost edge is blue.}
\label{f:core} \end{figure}

The following lemma shows that the conditions of Lemma~\ref{l:main} hold during the algorithm,
and thus the algorithm terminates.

\begin{lemma} \label{l:dH}
During the algorithm $d_H(u) \le 1$ and $d_I(u)+2d_H(u) \ge k-2$ 
holds for all~$u \in U$. % and $f(\{u\}) \le 0$
Furthermore, $d_H(v) \le 2$ for all $v \in V \sem U$.
\end{lemma}
\begin{proof}
By Lemma~\ref{l:C}, when $u$ enters $U$ we have 
$d_H(u) \le 1$ and  
$d_I(u) \ge k-2d_H(u)-2$.
This implies 
$f(\{u\})=k-d_I(u)-2d_H(u)-2 \le 0$.
The addition of a ghost edge incident to $u$ excludes $u$ from $U$ but does not increase $f(\{u\})$. 
Hence $\{u\}$ will never be chosen as a core $C$ and will never enter $U$ again. 
If we had $d_H(u)=1$ when $u$ entered $U$, then $u$ will leave $U$ with $d_H(u)=2$ 
and will never enter $U$ again. Thus $d_H(v) \le 2$ for all $v \in V \sem U$.
\end{proof}

A set family $\RR$ is {\bf laminar} if any two sets in the family are either disjoint or one contains the other. 
Let $\RR$ be the family of subsets of $V_0$ that correspond to the contracted sets during the algorithm
(recall that $G_0=(V_0,E_0)$ denotes the initial graph).
It is known and not hard to see that $\RR$ is laminar.
We will say that a cut $S$ of $G_0$ is {\bf compatible} with a graph $G$ during the algorithm
if $R \subs S$ or $R \cap S=\empt$ for every $R \in \RR$ that was contracted so far to obtain $G$. 
As long as $S$ is compatible with $G$, we can look at the cut of $G$ that corresponds to $S$,
which for simplicity of notation we will also denote by $S$, and lower bound 
$d_I(S)$ (the degree of the set that corresponds to $S$ w.r.t to the partial solution $I$) instead of 
$d_J(S)$ (the degree of $S$ w.r.t to the final solution $J$),
since $I \subs J$ and hence $d_J(S) \ge d_I(S)$.

We would like to establish the bound $d_J(S) \ge k-4$ for a cut $S$ of $G_0$ such that $\RR \cup \{S\}$ is laminar.
For such $S$, it is easy to determine the last graph $G$ that $S$ is compatible with.
For a cut $S$ of $G_0$ let $R_S$ be the inclusion minimal set in $\RR \cup \{V\}$ that properly contains $S$.
If $R_S=V$ then $S$ compatible with $G$ till the end of the algorithm.
If $R_S \ne V$ then $S$ (such that $\RR \cup \{S\}$ is laminar) 
is compatible with $G$ as long as $R_S$ is not contracted,
hence we will analyze such $S$ as a cut of $G$ right before the core $C=R_S$ is contracted.
Due to the fact that at this moment no edge in $E$ has both ends in $C$ 
and since $d_E(C) \le 3$ and $x(\de_E(C)) \le 2$ (by Lemma~\ref{l:main}(i)),
the difference between $d_I(S)$ and $d_J(S)$ is small. 

\begin{lemma} \label{l:cut}
When a core $C$ % at step $7$ of the algorithm 
in the algorithm is found but not yet contracted, 
the following holds for any proper subset $S$ of $C$.
\begin{enumerate}[(i)]
\item
$d_I(S) \ge d_H(S) \cdot \mu$ and $d_I(S) \ge k-2d_H(S)-2$.
\item
$d_I(S,C \sem S) \ge \mu$.
\end{enumerate}
\end{lemma}
\begin{proof}
We prove (i). $d_I(S) \ge d_H(S) \cdot \mu$ by Lemma~\ref{l:h}.
If $S=\{u\}$ for some $u \in U$ then $d_I(u) \ge k-2d_H(u)-2$ by Lemma~\ref{l:C}.
Otherwise, the constraint of $S$ is relaxed by ghost edges only, 
and since $\de_E(S) \subs \de_E(C)$ (Lemma~\ref{l:main}(i)) 
and $x(\de_E(C)) \le 2$ (Lemma~\ref{l:C}), we get 
\[
d_I(S) \ge k-2d_H(S) -x(\de_E(S)) \ge k-2d_H(S)-2 \ .
\]

We prove (ii). 
If $d_H(S,C \sem S) \ge 1$ then $d_I(S,C \sem S) \ge \mu$ 
since $d_I(u,v) \ge \mu$ when the ghost edge $uv$ is added. 
Otherwise, $d_H(S)+d_H(C \sem S)=d_H(C) \le 1$ 
and since $d_J(R) \le k-2d_H(C)-1$ (by Lemma~\ref{l:C}) 
and by part (i) we get 
\begin{eqnarray*}
2d_I(S,C \sem S) &   =   & d_I(S)+d_I(C \sem S)-d_I(C) \\
                              & \ge & [k-2d_H(S)-2] + [k-2d_H(C \sem S)-2]-[k-2d_H(C)-1] \\
															&   =  & k-3-2[d_H(S)+d_H(C \sem S)-d_H(C)] = k-3 \ ,
\end{eqnarray*}
concluding the proof.
\end{proof}

The notation $\de_H(S)$ and $d_H(S)$ are well defined for any node subset $S$ at any iteration of the algorithm. 
We now extend it to any subset $S$ of $V$. % such that $\RR \cup \{S\}$ is laminar. 
We can view a ghost edge $v_1v_2$ as a pair $\{C_1,C_2\}$ of sets in $\RR$
that correspond to $v_1,v_2$. 
We say that such a {\bf ghost-edge covers $S$} if one of $C_1,C_2$ is contained in $S$ and the other in $\bar{S}$.
Let $\de_H(S)$ denote the set of ghost edges in $H$ that cover $S$
and let $d_H(S)$ be their number.

\begin{corollary} \label{c:bound}
Let $S$ be a cut such that $\RR \cup \{S\}$ is laminar.
Then $d_J(S) \ge k-2d_H(S)-2$ and $d_J(S) \ge d_H(S) \cdot \mu$. 
Furthermore, if $R_S=V$ and $S \notin R$ then $d_J(S) \ge k-2d_H(S)$.
Consequently, $d_J(S) \ge \max\{k-2d_H(S)-2, d_H(S) \cdot \mu\} \ge k-4$.
\end{corollary}
\begin{proof}
The bound $d_J(S) \ge d_H(S) \cdot \mu$ follows from Lemma \ref{l:h}, 
so we prove only the bound $d_J(S) \ge k-2d_H(S)-2$.
For $S \in \RR$ it follows from Lemma \ref{l:C}, so assume that $S \notin \RR$. 
If $S$ is contained in some set in $\RR$, then let $C$ be the inclusion minimal set in $\RR$ that contains $S$.
The bound then follows from Lemma~\ref{l:cut}(i).
Otherwise, the constraint of $S$ was relaxed by ghost edges only and then $d_J(S) \ge k-2d_H(S)$.
\end{proof}

Before establishing the bound $d_J(S) \ge k-4$ for sets $S$ such that 
$\RR \cup \{S\}$ is not laminar, we need three additional lemmas.
The next lemma considers sets $R \in \RR$ and $T \subset \bar{R}$ 
such that both $\RR \cup T,\RR \cup (\bar{R} \sem T)$ are laminar;
one can verify that this is equivalent to $\RR \cup \{T,\bar{R} \sem T\}$ being laminar.
Moreover, this means that $R$ is an inclusion maximal set in $\RR$, and 
any other set in $\RR$, except for maybe $\bar{R}$ (if $\bar{R} \in \RR$), 
is contained in one of $R,T,\bar{R} \sem T$.  

\begin{lemma} \label{l:bR}
Let $R \in \RR$ and let $T$ be a proper subset of $\bar{R}$ such that the set family $\RR \cup \{T,\bar{R} \sem T\}$ is laminar.
Then $d_J(T,\bar{R} \sem T) \ge \lceil (k-6)/2 \rceil$.
\end{lemma}
\begin{proof}
If $d_H(T,\bar{R} \sem T) \ge 1$ then $d_J(T,R \sem T) \ge \mu>\lceil (k-6)/2 \rceil$ by Lemma~\ref{l:h}.
Otherwise, $d_H(R)=d_H(T)+d_H(\bar{R} \sem T)$. Since 
$d_J(T) \ge k-2d_H(T)-2$ and 
$d_J(\bar{R} \sem T) \ge k-2d_H(\bar{R} \sem T)-2$ by Corollary~\ref{c:bound}, and
$d_J(R) \le k-2d_H(R)+2$ by Lemma~\ref{l:C}, we get (see Fig.~\ref{f:count}(a))
\begin{eqnarray*}
2d_J(T,\bar{R} \sem T) &  =  & d_J(T)+d_J(\bar{R} \sem T)-d_J(R) \\
                                        & \ge & [k-2d_H(T)-2] + [k-2d_H(\bar{R} \sem T)-2]-[k-2d_H(R)+2] \\
															          &   =  & k-6-2[d_H(T)+d_H(\bar{R} \sem T)-d_H(R)] = k-6 \ ,
\end{eqnarray*}
concluding the proof.
\end{proof}

We say that a set $S$ {\bf overlaps} a set $R$, or that $R,S$ {\bf overlap},
if the pair $\{R,S\}$ is not laminar, namely if the sets $R \cap S, R \sem S, S \sem R$ are all non-empty.
Now we consider sets $S$ that overlap some $R \in \RR$.
Note however, that it may be that $S$ overlaps some $R \in \RR$, but $\bar{S}$ does not, 
so we can still deduce from Corollary~\ref{c:bound} that $d_J(S)=d_J(\bar{S}) \ge k-4$.
To avoid this confusion we will examine among the sets $S,\bar{S}$ one that is more ``compatible'' with $\RR$.
More formally, let $\RR(S)=\{R \in \RR: S \mbox{ overlaps } R\}$ be the family of sets in $\RR$ that $S$ overlaps,
and note that so far we examined the case $|\RR(S)|=0$. 
So now w.l.o.g. we will assume that $1 \le |\RR(S)| \le |\RR(\bar{S})|$. 
Under this assumption, we have the following.

\begin{lemma} 
If $1 \le |\RR(S)| \le |\RR(\bar{S})|$ then $V \sem (R \cup S) \ne \empt$ for any $R \in \RR(S)$. 
 \end{lemma}
\begin{proof}
Let $R \in \RR(S)$ and suppose that $|\RR(S)| \le |\RR(\bar{S})|$.
It is known that $\RR(R \sem S) \subset \RR(S)$;
this is so since every set in $\RR$ that overlaps $R \sem S$ also overlaps $S$, but $R$ overlaps $S$ but not $R \sem S$,
c.f.  \cite[Lemma~23.15]{V}. Thus $|\RR(R \sem S)| < |\RR(S)|$.
If $V \sem (R \cup S) = \empt$ then $R \sem S=\bar{S}$ and we get $|\RR(\bar{S})| < |\RR(S)|$
contradicting the assumption $|\RR(S)| \le |\RR(\bar{S})|$.
\end{proof}

\begin{lemma} \label{l:Rm}
Suppose that $1 \le |\RR(S)| \le |\RR(\bar{S})|$.
\begin{enumerate}[(i)]
\item
If $R$ is a minimal set in $\RR(S)$ then 
$d_J(R \cap S,R \sem S) \ge \mu$.
\item
If $R$ is a unique maximal set in $\RR(S)$ then 
$d_J(\bar{R} \cap S, \bar{R} \sem S) \ge \lceil (k-6)/2 \rceil$.
\end{enumerate}
\end{lemma}
\begin{proof}
We prove (i). The minimality of $R$ implies that if $T \in \{R \cap S,R \sem S\}$ 
then $\RR \cup \{T\}$ is laminar.
Therefore by Lemma \ref{l:cut}(ii) 
$d_J(R \cap S,R \sem S) \ge \mu$.

We prove (ii). The maximality of $R$ implies that if $T \in \{\bar{R} \cap S,\bar{R} \sem S\}$ 
then $\RR \cup \{T\}$ is laminar.
Therefore by Lemma \ref{l:bR} 
$d_J(\bar{R} \cap S,\bar{R} \sem S) \ge \lceil (k-6)/2 \rceil$.
\end{proof}

\begin{lemma} \label{l:count}
If $1 \le |\RR(S)| \le |\RR(\bar{S})|$ then $d_J(S) \ge k-4$. 
\end{lemma}
\begin{proof}
Suppose that there are two disjoint sets in $\RR(S)$. 
Then there are two minimal sets in $\RR(S)$ that are disjoint, 
say $R_1,R_2$; see Fig.~\ref{f:count}(b).
Then by Lemma~\ref{l:Rm}(i)
\[
d_J(S) \ge d_J(R_1 \cap S,R_1 \sem S)+ d_J(R_2 \cap S,R_2 \sem S) \ge \mu+\mu > k-4 \ .
\]
Otherwise, $\RR(S)$ is a nested family with a unique maximal set $R_1$ and a unique minimal set $R_2$; 
see Fig.~\ref{f:count}(c) and note that $R_2 \subs R_1$ and possibly $R_1=R_2$. 
Then by Lemma~\ref{l:Rm} we get
\[
d_J(S) \ge d_J(\bar{R}_1 \cap S,\bar{R}_1 \sem S)+ d_J(R_2 \cap S,R_2 \sem S) \ge \lceil (k-6)/2 \rceil + \mu=k-4 \ ,
\]
concluding the proof. 
\end{proof}

\begin{figure} \centering \includegraphics[scale=0.5]{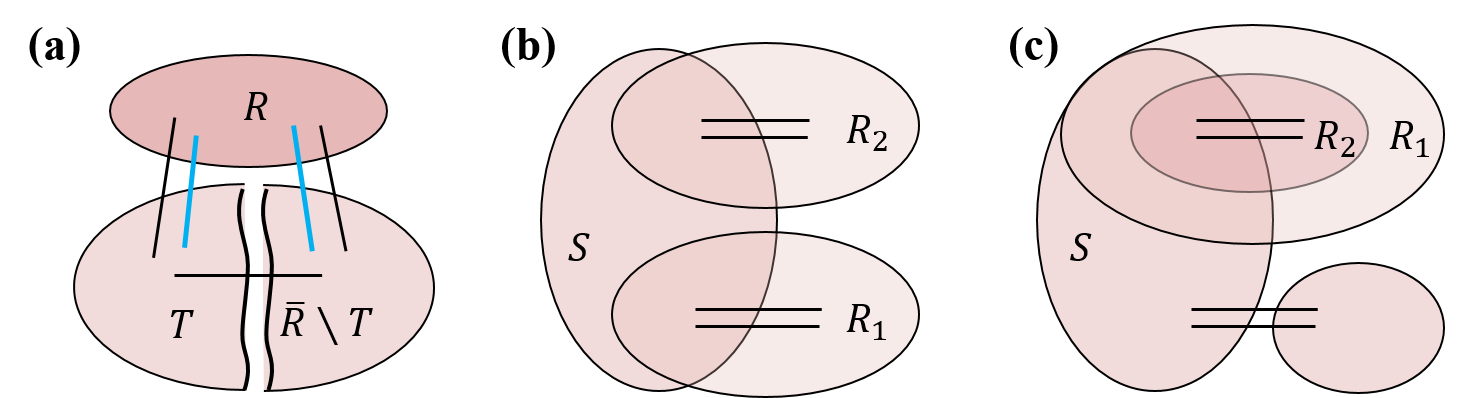}
\caption{Illustration to the proofs of Lemmas \ref{l:bR} and \ref{l:count}.}
\label{f:count} \end{figure}

This concludes the proof of the connectivity guarantee of Algorithm~\ref{alg:main}.
Now we show that the algorithm terminates after a polynomial number of iterations.

\begin{lemma} \label{l:m}
The algorithm terminates after $O(n)$ iterations, where $n=|V|$.
\end{lemma}
\begin{proof}
Let $x$ be an extreme point computed at the first iteration of the algorithm. 
It is known (c.f. \cite{GGTW}) that $|\{0<x_e <1:e \in E\}| \le 2n-1$, namely, at most $2n-1$ variables $x_e$ are fractional.
This follows from two facts. 
The first is that after the integral entries are substituted, the fractional entries are determined 
by a full rank system of constraints $\{x(\de_E(S))=k:S \in \LL\}$ where $\LL$ is laminar, c.f. \cite{J,GGTW}. 
The second is that a laminar family on a set of $n$ elements has size at most $2n-1$.
The latter fact also implies that $|\RR| \le 2n-1$.
By Lemma~\ref{l:dH}, the ghost edges form a $2$-matching on $\RR$, thus $|H| \le |\RR| \le 2n-1$. 
At each iteration we remove an edge from $E$, or add a set to $\RR$, or add a ghost edge.
Hence the number of iterations is bounded by 
$|\LL|+|\RR|+|H| \le (2n-1)+(2n-1)+(2n-1)=3(2n-1)=O(n)$.
\end{proof}

%%%%%%%%%%%%%%%%%%%%%%%%%%%%%
\section{Properties of extreme point solutions (Lemma~\ref{l:main})} \label{s:main}
%%%%%%%%%%%%%%%%%%%%%%%%%%%%%

Here we will prove Lemma~\ref{l:main}. 
Fix an extreme point solution $x$ as in the Lemma. 
Then there exists a family $\LL$ of $f$-positive $x$-tight sets such that 
$x$ is the unique solution to the equation system $\{x(\de_E(S))=f(S):S \in \LL\}$; namely, 
$|\LL|=|E|$ and the characteristic vectors of the sets $\{\de_E(S):S \in \LL\}$ are linearly independent.  
We will call such a family $\LL$ {\bf $x$-defining}.
The following was essentially proved in \cite{HKZ}.

\begin{lemma}[\cite{HKZ}] \label{l:L}
If there exists a laminar $x$-defining family $\LL$ then there is an $x$-core $C$ 
such that $d_E(C) \in \{2,3\}$ and $f(C) \in \{1,2\}$;
moreover, no edge in $E$ has both ends in $C$.
\end{lemma}
\begin{proof}
Let $\CC_\LL$ be the family of minimal sets in $\LL$. 
It was proved by HKZ \cite{HKZ} that there is $S \in \CC_\LL$ such that $d_E(S) \le 3$
and no edge in $E$ has both ends in $S$; 
we provide a proof sketch for completeness of exposition.
Suppose to the contrary that $d_E(S) \ge 4$ for every $S \in \CC_\LL$. 
Assign $2$ tokens to every edge $uv$, % (so the total amount of tokens is exactly $2|E|$), 
and then reassign these tokens -- 
one to the smallest set in $\LL$ that contains $u$ and 
one to the smallest set in $\LL$ that contains $v$ (if such sets exist). 
The paper of Jain \cite{J} shows that then for any $S \in \LL$, the tokens assigned to sets in 
$\LL(S)=\{T \in \LL:T \subseteq S\}$ can be redistributed such that $S$ gets at least $4$ tokens 
and every $T \in \LL(S)$ gets at least $2$ tokens, giving the contradiction $|E|>|\LL|$. 
The proof in \cite{J} is by induction on $|\LL(S)|$.
If $S \in \CC_\LL$ is a leaf then $S$ already has $4$ tokens, since $d_E(S) \ge 4$ for every $S \in \CC_\LL$. 
Else, $S$ can collect $2$ tokens from each child in $\LL(S)$, % (a maximal inclusion set in $\LL(S)$ properly contained in $S$), 
so if $S$ has at least two children then we are done. 
Else, $S$ has exactly one child, and then \cite{J} shows that by linear independence, 
there are $2$ tokens assigned to $S$, 
which together with the extra $2$ tokens of the child of $S$ gives the required $4$ tokens. 
This shows that there is $S \in \CC_\LL$ such that $d_E(S) \le 3$.

Note that for any $T \in \CC_\LL$, no edge in $E$ has both ends in $T$; 
this is since for every $e \in E$ there is $S_e \in \LL$ such that $e \in \de_E(S_e)$, 
since the variable $x_e$ must appear in at least one equation among $\{x(\de_E(S))=f(S):S \in \LL\}$.
Let $S \in \CC_\LL$ be such that $d_E(S) \le 3$.
If $S$ is an $x$-core then we are done. 
Otherwise, let $C$ be an $x$-core properly contained in $S$.
We claim that $\de_E(C) \subs \de_E(S)$, and that $d_E(C) \ge 2$.
This follows from the following observations.
\begin{enumerate}
\item
$d_E(C) \ge 2$, since $C$ is $f$-positive and since we assume that $0<x_e<1$ for all $e \in E$.
\item
For every $e \in \de_E(C)$ there is $S_e \in \LL$ such that $e \in \de_E(S_e)$, 
since the variable $x_e$ must appear in at least one equation among $\{x(\de(S))=f(S):S \in \LL\}$.
\end{enumerate}
This gives $d_E(C) \in \{2,3\}$, as required.
\end{proof}

In the rest of this section we will prove the following.

\begin{lemma} \label{l:no}
Under the assumptions of Lemma~\ref{l:main},
if no laminar $x$-defining family exists then there are $u,v \in U$ with $d_I(u,v) \ge \mu$ and $d_H(u,v)=0$.
\end{lemma}

In the rest of this section assume that the assumptions of Lemma~\ref{l:main} hold 
($d_H(u) \le 1$ and $d_I(u)+2d_H(u) \ge k-2$ for all $u \in U$),
and that no laminar $x$-defining family exists.

\begin{lemma} \label{l:AB}
% If no laminar $x$-defining family exists then 
There are $f$-positive $x$-tight sets $A,B$ that cross 
(namely, all the four sets $A \cap B, A \sem B, B \sem A, \bar{A} \cap \bar{B}$ are non-empty) 
such that the following holds:
\begin{enumerate}[(i)]
\item
$A \cap B=\{u\}$ or $\bar{A} \cap \bar{B}=\{u\}$ for some $u \in U$. 
% (namely, at least one of $A \cap B,\bar{A} \cap \bar{B}$ is a singleton from $U$).
\item
$A \sem B=\{v\}$ or $B \sem A=\{v\}$ for some $v \in U$.
% (namely, at least one of $A \sem B,B \sem A$ is a singleton from $U$).
\end{enumerate}
\end{lemma}
\begin{proof}
For a cut $S$ let $\chi_S \in \{0,1\}^E$ be the incidence vector of $\de_E(S)$.
Let $\TT$ be the family of $f$-positive $x$-tight sets
(if $S$ is tight but not $f$-positive, then $\chi_S \equiv 0$, since $x_e>0$ for all $e \in E$).
Let us say that $A,B \in \TT$ are {\bf $x$-uncrossable} if at least one of the following holds:
\begin{enumerate}[(a)]
\item
$A \cap B,A \cup B$ are both $x$-tight and $\chi_A+\chi_B=\chi_{A \cap B}+\chi_{A \cup B}$. 
\item
$A \sem B,B \sem A$ are both $x$-tight and  $\chi_A+\chi_B=\chi_{A \sem B}+\chi_{B \sem A}$. 
\end{enumerate}
We will prove two claims that will imply the lemma.

\begin{claim*}
There are $A,B \in \TT$ that are not $x$-uncrossable.
\end{claim*}
{\em Proof.}
Let $\LL$ be a maximal laminar subfamily of $\TT$.
Let $\span(\TT)$ denote the linear space spanned by the incidence vectors of the sets in $\TT$ 
and similarly $\span(\LL)$ is defined. 
Since no laminar $x$-defining family exists and by the maximality of $\LL$, 
there is $A \in \TT \sem \LL$ such that $\chi_A \notin \span(\LL)$;
any such $A$ overlaps some set in $\LL$.
Choose such $A$ that overlaps the minimal number of sets in $\LL$, 
and let $B \in \LL$ be a set that $A$ overlaps.
We show that (a) cannot hold; the proof that (b) cannot hold is similar. 
Suppose to the contrary that (a) holds, namely that 
$\chi_A+\chi_B=\chi_{A \cap B}+\chi_{A \cup B}$ and
$A \cap B,A \cup B$ are both $x$-tight. % and at least one of them is $f$-positive. 
It is known that each of the sets $A \cap B,A \cup B$ overlaps strictly less sets in $\LL$ than $A$, c.f. \cite[Lemma~23.15]{V}.
Thus by our choice of $A$, each of these sets has its incidence vector in $\span(\LL)$, namely $\chi_{A \cap B}, \chi_{A \cup B} \in \span(\LL)$.
But $\chi_A=\chi_{A \cap B}+\chi_{A \cup B}-\chi_B$, contradicting that $\chi_A \notin \span(\LL)$.
\hfill $\Box$

\begin{claim*}
If $A,B \in \TT$ are not $x$-uncrossable then $A,B$ cross and both (i) and (ii) hold.
\end{claim*}
{\em Proof.}
We show that $A,B$ cross. If $A,B$ do not overlap then
$\{A,B\}=\{A \cap B,A \cup B\}$ or $\{A,B\}=\{A \sem B,B \sem A\}$ and we are done.
We claim that also $\bar{A} \cap \bar{B} \ne \empt$.
Otherwise, $A \sem B=\bar{B}$, $B \sem A=\bar{A}$, 
and thus $\chi_{A \sem B}=\chi_{\bar{B}}=\chi_B$ and $\chi_{B \sem A}=\chi_{\bar{A}}=\chi_A$,
implying that $A \sem B,B \sem A$ are both $x$-tight and  $\chi_A+\chi_B=\chi_{A \sem B}+\chi_{B \sem A}$. 
Now we prove that (i) holds; a proof that (ii) holds is similar, and in fact can be deduced from (i) by the symmetry of $f$. 
Suppose to the contrary that (i) does not hold. 
Denoting $\De=x(\de_E(A \sem B,B \sem A))$ we have:
\begin{eqnarray*}
f(A)+f(B) & = & x(\de_E(A))+x(\de_E(B)) \\
               & = & x(\de_E(A \cap B))+x(\de_E(A \cup B))+2\De \\
						& \ge & f(A \cap B)+f(A \cup B)+2 \De \\
						& \ge & f(A)+f(B)+2\De \ .
\end{eqnarray*}
The first equality is since $A,B$ are tight, and 
the second can be verified by counting the contribution of each edge to both sides. 
The first inequality is since $x$ is a feasible LP-solution. 
The second inequality is since none of $A,B,A \cap B,A \cup B$ is a singleton from $U$ or its complement,
hence $f$ coincides on these sets with the symmetric supermodular function $g(S)=k-d_I(S)-2d_H(S)$,
that satisfies the inequality $g(A \cap B)+g(A \cup B) \ge g(A)+g(B)$. 
Consequently, equality holds everywhere, hence $A \cap B,A \cup B$ are both tight.
Moreover, $\De=0$, and this implies $\chi_A+\chi_B=\chi_{A \cap B}+\chi_{A \cup B}$,
contradicting that $A,B$ are not $x$-uncrossable.
\hfill $\Box$

\medskip

By the first claim there exists $A,B \in \TT$ that are not $x$-uncrossable, 
while by the second claim such $A,B$ cross and satisfy both (i) and (ii), 
concluding the proof of the lemma.
\end{proof}

By the symmetry of $f$, assume w.l.o.g. that $A \cap B=\{u\}$ and $A \sem B=\{v\}$ for $u,v \in U$.

\begin{lemma} \label{l:mu}
If $d_H(u,v)=0$ then $d_I(u,v) \ge \mu$.
\end{lemma}
\begin{proof}
By the assumption of Lemma~\ref{l:main}, $d_I(u)+2d_H(u) \ge k-2$ and $d_I(v)+2d_H(v) \ge k-2$.
Since $A$ is $f$-positive, $d_I(A)+2d_H(A) \le k-1$. Thus we get (see Fig.~\ref{f:w}(a))
\[
k-1 \ge d_I(A)+2d_H(A)=[d_I(u)+2d_H(u)]+[d_I(v)+2d_H(v)]-2d_I(u,v) \ge 2(k-2)-2d_I(u,v) \ .
\]
Consequently, $2d_I(u,v) \ge k-3$, as claimed.
\end{proof}

\begin{figure} \centering \includegraphics[scale=0.5]{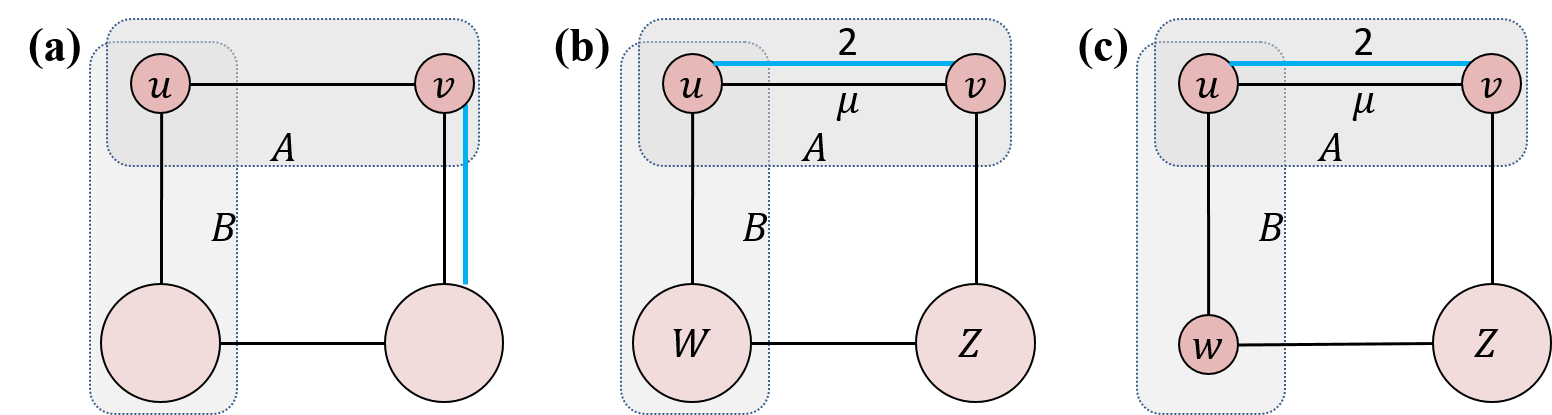}
\caption{Illustration to the proofs of Lemmas \ref{l:mu} and \ref{l:w}.}
\label{f:w} \end{figure}

The next lemma shows that if $d_H(u,v)=1$ then there is an ``alternative'' node $w \in U$ 
such that we can add a ghost edge $uw$ or $vw$. 

\begin{lemma} \label{l:w}
If $d_H(u,v)=1$ then there is $w \in U$ such that one of the following holds:
\begin{enumerate}[(i)]
\item
$\bar{A} \cap B=\{w\}$, $d_H(u,w)=0$, and $d_I(u,w) \ge \mu$.
\item
$\bar{A} \sem B=\{w\}$, $d_H(v,w)=0$, and $d_I(v,w) \ge \mu$.
\end{enumerate}
\end{lemma}
\begin{proof}
For disjoint sets $S,T$ let 
$\psi(S,T)=d_I(S,T)+2d_H(S,T)+x(\de_E(S,T))$, 
and let $\psi(S)=\psi(S,\bar{S})$ be the ``coverage'' of $S$ by $I$, $H$, and the fractional edges.
Note that
\[
\psi(S) \ge \left \{ \begin{array}{ll}
k-2 \ \  & \mbox{if } S=\{u\} \mbox{ or } \bar{S}=\{u\} \mbox{ for some } u \in U \\
k          & \mbox{otherwise}
\end{array} \right .
\]
In particular, $\psi(S)<k$ implies that $S=\{w\}$ or $\bar{S}=\{w\}$ for some $w \in U$.

Let $W=\bar{A} \cap B$ and $Z=\bar{A} \sem B$, see Fig.~\ref{f:w}(b).
Note that $\psi(A)=\psi(B)=k$ since $A,B$ are $x$-tight and that 
$\psi(u,v) \ge 2d_H(u,v)+ \mu = 2+\mu$. 
This implies (see Fig.~\ref{f:w}(b))
\[
\psi(W,Z) \le \psi(B)-\psi(u,v) \le k-(2+\mu) \ .
\] 
Thus we get 
\[
2k-2(2+\mu) \ge 2\psi(W,Z)=\psi(W)+\psi(Z)- \psi(A) = \psi(W)+\psi(Z)-k
\]
Consequently, $\psi(W)+\psi(Z) \le 3k-2(2+\mu) <2k$, since $2(2+\mu)>k$.  
Thus $\psi(W)<k$ or $\psi(Z)<k$, and w.l.o.g. assume that $\psi(W)<k$. 
This implies that $W=\{w\}$ for some $w \in U$; (see Fig.~\ref{f:w}(c)).
Note that $d_H(u,w)=0$, since $d_H(u,v)=1$ and since $d_H(u) \le 1$ by the assumption of Lemma~\ref{l:main}.
Also, $d_I(B)+2d_H(B) \le k-1$, since $B$ is $f$-positive.
Thus we have
\begin{eqnarray*}
2d_I(u,w) & = & 2[d_I(u,w)+2d_H(u,w)] \\
                 &  =  & [d_I(u)+2d_H(u)]+[d_I(w)+2d_H(w)]-[d_I(B)+2d_H(B)] \\
                 & \ge & (k-2)+(k-2)-(k-1) =k-3 \ ,
\end{eqnarray*}
concluding the proof.
\end{proof}

To show a polynomial time implementation it is sufficient to prove the following lemma, 
that was essentially proved in \cite{HKZ}; we provide a proof for completeness of exposition.

\begin{lemma}[\cite{HKZ}] \label{l:impl}
At any step of the algorithm, an extreme point solution $x$ to LP  (\ref{LP}) with $f$ in (\ref{f}), as well
as the $x$-cores can be computed in polynomial time. 
\end{lemma}
\begin{proof}
An extreme point solution can be found using the Ellipsoid Algorithm.
For that, we need the show an existence of a polynomial time separation oracle: 
given $x$, determine whether $x$ is a feasible LP solution or find a violated inequality.
Checking the inequalities $0 \le x_e \leq 1$ is trivial.
For cut inequalities, assign capacities $w_e$ to the edges in $E \cup I \cup H$ as follows: 
$w_e=x_e$ if $e \in E$, $w_e=1$ if $e \in I$, and $w_e=2$ if $e \in H$. 
The inequality of a cut $S$ is violated if and only if 
$w(\de(S)) < k$ and none of $S,\bar{S}$ is a singleton in $U$. 
Thus to find an $st$-cut with violated inequality do the following.
If $s,t \notin U$ then we find a minimum $st$-cut $S$;
if $w(\de(S)) < k$ then we found a violated inequality, else no $st$-cut with violated inequality exists.
If $s \in U$ and $t \notin U$, then for every $u \in V \sem \{s,t\}$ we contract $u$ into $s$ 
and find a minimum $st$-cut $S_v$; 
if $w(\de(S_v)) < k$ for some $v$, then we found a violated inequality, else no $st$-cut with violated inequality exists.
Similarly, If $s,t \in U$ then for every pair $u,v \in V \sem \{s,t\}$ we contract $u$ into $s$ and $v$ into $t$
and find a minimum $st$-cut $S_{u,v}$; 
if $w(\de(S_{u,v})) < k$ for some $u,v$ then we found a violated inequality, else no $st$-cut with violated inequality exists.

We show how to find all $x$-cores. Assume w.l.o.g. that $0 < x_e <1$ for all $e \in E$.
One can see that a cut $C$ is an $x$-core if and only if 
$C$ is an inclusion minimal set that has $w$-capacity $k$ and $\de_E(C) \ne \empt$. 
For every ordered pair $(s,t)$ with $st \in E$ we check if there exists an $st$-cut $C$ (so $s \in C$ and $t \notin C$)
of capacity exactly $k$, and if so, then we find the inclusion minimal such cut $C(s,t)$.
The $x$-cores are the inclusion minimal sets among the sets $C(s,t)$.
\end{proof}

This concludes the proof of Lemma~\ref{l:main} and thus also of Theorem~\ref{t:main}.

%%%%%%%%%%%%%%%%%%%%%%%%%%%%%%%
\section{Proof of Theorem~\ref{t:main'}} \label{s:main'}
%%%%%%%%%%%%%%%%%%%%%%%%%%%%%%%

In the algorithm of Theorem~\ref{t:main'} we use the following function $f$:
\begin{equation} \label{f'}
f(S) = \left \{ \begin{array}{ll}
k-d_I(S)-d_H(S)-1 \ \ & \mbox{if } S=\{u\} \mbox{ or } \bar{S}=\{u\} \mbox{ for some } u \in U \\
k-d_I(S)-d_H(S)         & \mbox{otherwise}
\end{array} \right .
\end{equation}

The main difference between Algorithm~\ref{alg:main} and the following Algorithm~\ref{alg:main'} are:
\begin{enumerate}
\item
We use a different function $f$ -- the one defined in (\ref{f'}). % (line 3 in the algorithm).
\item
We round to $1$ edges $e$ with $x_e \ge 2/3$. % (lines 5,6 in the algorithm).
\item
Adding a ghost edge $uv$ requires  
$\mu=\lceil (k-1)/2 \rceil$ integral $uv$-edges. % (line 9 in the algorithm).
\end{enumerate}

\begin{algorithm}[H]
\caption{A bicriteria $(3/2,k-2)$-approximation} \label{alg:main'}
{\bf initialization}: $I \gets \empt$, $H \gets \empt$, $U \gets \empt$, $\mu \gets \lceil (k-1)/2 \rceil$ \\
\While{$E \ne \empt$}
{
compute an extreme point solution $x$ to LP  (\ref{LP}) with $f$ in (\ref{f'}) \\
remove from $E$ every edge $e$ with $x_e=0$ \\
\If{\em there is $e \in E$ with $x_e \ge 2/3$}
{move from $E$ to $I$ every edge $e$ with $x_e \ge 2/3$ and goto line 2} 
\ElseIf{\em there is an $x$-core $C$ with $d_I(C)=k-1$ and $d_E(C) \in \{2,3\}$}
{contract $C$ into a single node $v_C$ and update $U \gets (U \sem C) \cup \{v_C\}$} 
\ElseIf{\em there are $u,v \in U$ with $d_I(u,v) \ge \mu$ and $d_H(u,v) =0$}  
{add to $H$ a ghost-edge $uv$ and update $U \gets U \sem \{u,v\}$}
}
\Return{$J=I$}
\end{algorithm}

\medskip \medskip

Note that here we combine iterative relaxation with iterative rounding.
Since in each iteration the cut constraints are only relaxed, 
and only edges $e$ with $x_e \ge 2/3$ are rounded,
the solution cost is at most $3/2$ times the initial LP-value.
It remains to prove a polynomial running time and the connectivity guarantee $k-2$.
To show that the algorithm terminates and runs in polynomial time 
it is sufficient to prove the following counterpart of Lemma~\ref{l:main}. 

\begin{lemma} \label{l:main'}
Let $f$ be defined by (\ref{f'}) and suppose that
$d_H(u) \le 1$ and $d_I(u)+d_H(u) \ge k-1$ for all $u \in U$.
Let $x$ be an extreme point of the polytope 
$P=\{x \in [0,1]^E:x(\de_E(S)) \ge f(S)\}$ % of the LP in (\ref{LP})
with $0<x_e<2/3$ for all $e \in E$, where $E \ne \empt$. 
Then at least one of the following holds.
\begin{enumerate}[(i)]
\item
There is an $x$-core $C$ with $d_E(C) \in \{2,3\}$ and $f(C)=1$; moreover, no edge in $E$ has both ends in $C$.
\item
There are $u,v \in U$ with $d_H(u,v)=0$ and $d_I(u,v) \ge \mu$, where $\mu=\lceil (k-1)/2 \rceil$.
\end{enumerate}
Furthermore, such $u,v$ or such $C$ can be found in polynomial time.
\end{lemma}
\begin{proof}
By Lemma~\ref{l:L}, if there exists an $x$-defining laminar family, then 
there is an $x$-core $C$ such that $d_E(C) \in \{2,3\}$, $f(C) \in \{1,2\}$, and no edge in $E$ has both ends in $C$.
But $f(C) = 2$ is not possible, as otherwise there is $e \in \de_E(C)$ with $x_e \ge 2/3$.
Hence in this case (i) holds.

\medskip

We now show that if no laminar $x$-defining family exists then (ii) holds.
Similarly to Lemma~\ref{l:AB} we conclude that there is a pair $A,B$ of $x$-tight $f$-positive 
sets that cross such that 
$A \cap B=\{u\}$ and $A \sem B=\{v\}$ for some $u,v \in U$.

We claim that if $d_H(u,v)=0$ then $d_I(u,v) \ge \mu$.
By the assumption of the lemma, $d_I(u)+d_H(u) \ge k-1$ and $d_I(v)+d_H(v) \ge k-1$.
Since $A$ is $f$-positive, $d_I(A)+d_H(A) \le k-1$. Thus we get (see Fig.~\ref{f:w}(a))
\[
k-1 \ge d_I(A)+d_H(A)=[d_I(u)+d_H(u)]+[d_I(v)+d_H(v)]-2d_I(u,v) \ge 2(k-1)-2d_I(u,v) \ .
\]
Consequently, $2d_I(u,v) \ge k-1$, as claimed.

\medskip

We prove that if $d_H(u,v)=1$ then there is $w \in U$ such that one of the following holds:
\begin{itemize}
\item
$\bar{A} \cap B=\{w\}$, $d_H(u,w)=0$, and $d_I(u,w) \ge \mu$.
\item
$\bar{A} \sem B=\{w\}$, $d_H(v,w)=0$, and $d_I(v,w) \ge \mu$.
\end{itemize}

% For disjoint sets $S,T$ 
Let $\psi(S,T)=d_I(S,T)+d_H(S,T)+x(\de_E(S,T))$
and $\psi(S)=\psi(S,\bar{S})$.
Note that
\[
\psi(S) \ge \left \{ \begin{array}{ll}
k-1 \ \  & \mbox{if } S=\{u\} \mbox{ or } \bar{S}=\{u\} \mbox{ for some } u \in U \\
k          & \mbox{otherwise}
\end{array} \right .
\]
In particular, $\psi(S)<k$ implies that $S=\{w\}$ or $\bar{S}=\{w\}$ for some $w \in U$.

Let $W=\bar{A} \cap B$ and $Z=\bar{A} \sem B$. 
Note that $\psi(A)=\psi(B)=k$ since $A,B$ are $x$-tight and that 
$\psi(u,v) \ge d_H(u,v)+ \mu = 1+\mu$. 
This implies 
$\psi(W,Z) \le \psi(B)-\psi(u,v) \le k-(1+\mu)$, 
see Fig.~\ref{f:w}(b).
Thus we get 
\[
2k-2(1+\mu) \ge 2\psi(W,Z)=\psi(W)+\psi(Z)- \psi(A) = \psi(W)+\psi(Z)-k
\]
Consequently, $\psi(W)+\psi(Z) \le 3k-2(1+\mu) <2k$, since $2(1+\mu)>k$.  
Thus $\psi(W)<k$ or $\psi(Z)<k$, and w.l.o.g. $\psi(W)<k$. 
This implies that $W=\{w\}$ for some $w \in U$, see Fig.~\ref{f:w}(c).
Note that $d_H(u,w)=0$, since $d_H(u,v)=1$ and since $d_H(u) \le 1$ by the assumption of the lemma.
Also, $d_I(B)+d_H(B) \le k-1$, since $B$ is $f$-positive.
Thus we have
\begin{eqnarray*}
2d_I(u,w) & = & 2[d_I(u,w)+d_H(u,w)] \\
                 &  =  & [d_I(u)+d_H(u)]+[d_I(w)+d_H(w)]-[d_I(B)+d_H(B)] \\
                 & \ge & (k-1)+(k-1)-(k-1) =k-1 \ .
\end{eqnarray*}

The polynomial time implementation details are identical to those in Lemma~\ref{l:impl}.
\end{proof}

\medskip

We now prove the connectivity guarantee. 
As before, let $\RR$ be the family of subsets of $V_0$ that correspond to the contracted sets during the algorithm.
Similarly to the previous case, we have the following sequence of lemmas. 

\begin{lemma}\label{l:C'}
When a core $C$ % at step $7$ of the algorithm 
is found in the algorithm but not yet contracted, 
$d_H(C) \le 1$, $x(\de_E(C)) =1$, and $d_I(C) = k-d_H(C)-x(\de_E(C))=k-d_H(C)-1$. 
Consequently, $d_I(C)= k-1$ if $d_H(C)=0$ and $d_I(C)=k-2$ if $d_H(C)=1$.
\end{lemma}
\begin{proof}
If $d_H(C) \ge 2$ then $d_I(C)+2d_H(C) \ge 2\mu +2 \ge k$, 
contradicting that $C$ is a core.  
Since $d_E(C) \in \{2,3\}$ and $x(\de_E(C))$ is a positive integer, we must have 
$x(\de_E(C))=1$, as otherwise there is $e \in \de_E(C)$ with $x_e \ge 2/3$. 
\end{proof}

\begin{lemma} \label{l:dH'}
During the algorithm $d_H(u) \le 1$ and $d_I(u)+d_H(u) = k-1$ 
holds for all~$u \in U$. Furthermore, $d_H(v) \le 2$ for all $v \in V \sem U$.
Thus the algorithm terminates after $O(m)$ iterations.
\end{lemma}
\begin{proof}
By Lemma~\ref{l:C'}, when $u$ enters $U$ we have 
$d_H(u) \le 1$ and  
$d_I(u) = k-d_H(u)-1$.
This implies 
$f(\{u\})=k-d_I(u)-d_H(u)-1 = 0$.
The addition of a ghost edge incident to $u$ excludes $u$ from $U$ but does not increase $f(\{u\})$, 
hence $\{u\}$ will never enter $U$ again. 
If we had $d_H(u)=1$ when $u$ entered $U$, then $u$ will leave $U$ with $d_H(u)=2$ 
and will never enter $U$ again. Thus $d_H(v) \le 2$ for all $v \in V \sem U$.
The proof that the algorithm terminates after $O(m)$ iterations is identical to the proof of Lemma~\ref{l:m}.
\end{proof}

\begin{lemma} \label{l:cut'}
When a core $C$ % at step $7$ of the algorithm 
is found in the algorithm but not yet contracted, 
the following holds for any proper subset $S$ of $C$.
\begin{enumerate}[(i)]
\item
$d_I(S) \ge d_H(S) \cdot \mu$ and $d_I(S) \ge k-d_H(S)-1$.
\item
$d_I(S,C \sem S) \ge \mu$.
\end{enumerate}
\end{lemma}
\begin{proof}
We prove (i). $d_I(S) \ge d_H(S) \cdot \mu$ by Lemma~\ref{l:h}.
If $S=\{u\}$ for some $u \in U$ then $d_I(u) = k-d_H(u)-1$ by Lemma~\ref{l:C'}.
Otherwise, the constraint of $S$ is relaxed by ghost edges only, 
and since $\de_E(S) \subs \de_E(C)$ (by Lemma~\ref{l:main'}(i)) 
and $x(\de_E(C)) \le 1$ (by Lemma~\ref{l:C'}), we get 
$d_I(S) \ge k-d_H(S) -x(\de_E(S)) = k-d_H(S)-1$.

We prove (ii). 
If $d_H(S,C \sem S) \ge 1$ then $d_I(S,C \sem S) \ge \mu$ 
since $d_I(u,v) \ge \mu$ when the ghost edge $uv$ is added. 
Otherwise, $d_H(S)+d_H(C \sem S)=d_H(C) \le 1$ and since $d_J(C) \le k-d_H(C)-1$ (by Lemma~\ref{l:C'}) 
and by part (i) we get 
\begin{eqnarray*}
2d_I(S,C \sem S) &   =   & d_I(S)+d_I(C \sem S)-d_I(C) \\
                              & \ge & [k-d_H(S)-1] + [k-d_H(C \sem S)-1]-[k-2d_H(C)-1] \\
															&   =  & k-1-[d_H(S)+d_H(C \sem S)-d_H(C)] = k-1 \ ,
\end{eqnarray*}
concluding the proof.
\end{proof}

\begin{corollary} \label{c:bound'}
Let $S$ be a cut such that $\RR \cup \{S\}$ is laminar.
Then $d_J(S) \ge k-d_H(S)-1$ and $d_J(S) \ge d_H(S) \cdot \mu$. 
Furthermore, if $R_S=V$ and $S \notin R$ then $d_J(S) \ge k-d_H(S)$.
Consequently, $d_J(S) \ge \max\{k-d_H(S)-1, d_H(S) \cdot \mu\} \ge k-2$.
\end{corollary}
\begin{proof}
The bound $d_J(S) \ge d_H(S) \cdot \mu$ follows from Lemma \ref{l:h}, 
so we prove only the bound $d_J(S) \ge k-d_H(S)-1$.
For $S \in \RR$ it follows from Lemmas \ref{l:C'}, so assume that $S \notin \RR$. 
If $S$ is contained in some set in $\RR$, then let $C$ be the inclusion minimal set in $\RR$ that contains $S$.
The bound then follows from Lemma~\ref{l:cut'}(i).
Otherwise, the constraint of $S$ was relaxed by ghost edges only and then $d_J(S) \ge k-d_H(S)$.
\end{proof}

\begin{lemma} \label{l:bR'}
Let $R \in \RR$ and let $T$ be a proper subset of $\bar{R}$ such that $\RR \cup \{T,\bar{R} \sem T\}$ is laminar.
Then $d_J(T,\bar{R} \sem T) \ge \lceil (k-4)/2 \rceil$.
\end{lemma}
\begin{proof}
If $d_H(T,\bar{R} \sem T) \ge 1$ then by Lemma~\ref{l:h} $d_J(T,R \sem T) \ge \mu>\lceil (k-4)/2 \rceil$.
Otherwise, $d_H(R)=d_H(T)+d_H(\bar{R} \sem T)$, and since $d_J(R) \le k-d_H(C)+2$ (by Lemma~\ref{l:C'}) we get 
\begin{eqnarray*}
2d_J(T,\bar{R} \sem T) &  =  & d_J(T)+d_J(\bar{R} \sem T)-d_J(R) \\
                                         & \ge & [k-d_H(T)-1] + [k-d_H(\bar{R} \sem T)-1]-[k-d_H(R)+2] \\
															           &   =  & k-4-[d_H(T)+d_H(\bar{R} \sem T)-d_H(R)] = k-4 \ ,
\end{eqnarray*}
concluding the proof.
\end{proof}

\begin{lemma} \label{l:Rm'}
Suppose that $1 \le |\RR(S)| \le |\RR(\bar{S})|$.
\begin{enumerate}[(i)]
\item
If $R$ is a minimal set in $\RR(S)$ then 
$d_J(R \cap S,R \sem S) \ge \mu$.
\item
If $R$ is a unique maximal set in $\RR(S)$ then 
$d_J(\bar{R} \cap S, \bar{R} \sem S) \ge \lceil (k-4)/2 \rceil$.
\end{enumerate}
\end{lemma}
\begin{proof}
We prove (i). The minimality of $R$ implies that if $T \in \{R \cap S,R \sem S\}$ 
then $\RR \cup \{T\}$ is laminar.
Therefore by Lemma \ref{l:cut'}(ii) 
$d_J(R \cap S,R \sem S) \ge \mu$.

We prove (ii). The maximality of $R$ implies that if $T \in \{\bar{R} \cap S,\bar{R} \sem S\}$ 
then $\RR \cup \{T\}$ is laminar.
Therefore by Lemma \ref{l:bR'} 
$d_J(\bar{R} \cap S,\bar{R} \sem S) \ge \lceil (k-4)/2 \rceil$.
\end{proof}

\begin{lemma} \label{l:count'}
If $1 \le |\RR(S)| \le |\RR(\bar{S})|$ then $d_J(S) \ge k-2$. 
\end{lemma}
\begin{proof}
Suppose that there are two disjoint sets in $\RR(S)$. 
Then there are two minimal sets in $\RR(S)$ that are disjoint, 
say $R_1,R_2$; see Fig.~\ref{f:count}(b).
Then by Lemma~\ref{l:Rm'}(i)
\[
d_J(S) \ge d_J(R_1 \cap S,R_1 \sem S)+ d_J(R_2 \cap S,R_2 \sem S) \ge \mu+\mu > k-2 \ .
\]
Otherwise, $\RR(S)$ is a nested family with a unique maximal set $R_1$ and a unique minimal set $R_2$; 
see Fig.~\ref{f:count}(c) and note that $R_2 \subs R_1$ and possibly $R_1=R_2$. 
Then by Lemma~\ref{l:Rm'}
\[
d_J(S) \ge d_J(\bar{R}_1 \cap S,\bar{R}_1 \sem S)+ d_J(R_2 \cap S,R_2 \sem S) \ge \lceil (k-4)/2 \rceil + \mu=k-2 \ ,
\]
concluding the proof. 
\end{proof}

This concludes the proof of the connectivity guarantee of Algorithm~\ref{alg:main'} 
and thus also of Theorem~\ref{t:main'}.

%%%%%%%%%%%%%%%
% \bibliographystyle{plainurl}
% \bibliography{kECSS-bib}
%%%%%%%%%%%%%%%

\end{document}